\documentclass[aps,pra,letterpaper,twocolumn]{revtex4}

\usepackage{amsmath,amsthm,amsfonts,graphicx,enumerate}

\begin{document}

\newcommand{\ket}[1]{\left| #1 \right\rangle}
\newcommand{\bra}[1]{\left\langle #1 \right|}
\newcommand{\braket}[2]{\left\langle #1 | #2 \right\rangle}
\newcommand{\braopket}[3]{\bra{#1}#2\ket{#3}}
\newcommand{\proj}[1]{| #1\rangle\!\langle #1 |}
\newcommand{\expect}[1]{\left\langle#1\right\rangle}
\newcommand{\Tr}{\mathrm{Tr}}
\def\Id{1\!\mathrm{l}}
\newcommand{\pmat}[1]{\begin{pmatrix}#1\end{pmatrix}}
\newcommand{\diff}{\mathrm{d}\!}
\newcommand{\complex}{\mathbb{C}}
\newcommand{\cT}{\mathcal{T}}
\newcommand{\cL}{\mathcal{L}}

\def\FPW{0.9\textwidth}
\def\HPW{0.45\textwidth}

\newtheorem{theorem}{Theorem}
\newtheorem{lemma}{Lemma}
\newtheorem{example}{Example}
\newtheorem{note}{Note}
\newtheorem{protocol}{Protocol}

\graphicspath{{Figures/}}

\bibliographystyle{apsrev}

\title{Streaming universal distortion-free entanglement concentration}

\author{Robin Blume-Kohout}
\email{robin@blumekohout.com}
\author{Sarah Croke}
\email{scroke@perimeterinstitute.ca}
\author{Daniel Gottesman}
\email{dgottesman@perimeterinstitute.ca}
\affiliation{Perimeter Institute}

\begin{abstract}
This paper presents a streaming (sequential) protocol for universal entanglement concentration at the Shannon bound.  Alice and Bob begin with $N$ identical (but unknown) two-qubit pure states, each containing $E$ ebits of entanglement.  They each run a reversible algorithm on their qubits, and end up with $Y$ perfect EPR pairs, where $Y = NE \pm O(\sqrt N)$.  Our protocol is streaming, so the $N$ input systems are fed in one at a time, and perfect EPR pairs start popping out almost immediately.  It matches the optimal block protocol exactly at each stage, so the average yield after $n$ inputs is $\expect{Y} = nE - O(\log n)$.  So, somewhat surprisingly, there is no tradeoff between yield and lag -- our protocol optimizes both.  In contrast, the optimal $N$-qubit block protocol achieves the same yield, but since no EPR pairs are produced until the entire input block is read, its lag is $O(N)$.  Finally, our algorithm runs in $O(\log N)$ space, so a lot of entanglement can be efficiently concentrated using a very small (e.g., current or near-future technology) quantum processor.  Along the way, we find an optimal streaming protocol for extracting randomness from classical i.i.d. sources and a more space-efficient implementation of the Schur transform.
\end{abstract}

\maketitle

Entanglement between two distant parties is an essential ingredient in quantum communication primitives such as teleportation \cite{BennettPRL93} and dense coding \cite{BennettPRL92}.  It is fungible, and can be transformed with negligible loss between different bipartite states, but the standard currency is EPR pairs, two-qubit states of the form
\begin{equation}
\ket{\Phi^+} = \frac{\ket{0_A0_B}+\ket{1_A1_B}}{\sqrt2},
\end{equation}
where the separated parties ``Alice'' and ``Bob'' each possess one qubit.  Most information processing protocols that use entanglement are designed to use perfect EPR pairs, so if the parties have some generic entangled state $\rho_{AB}$, their first order of business is to transform it into EPR pairs.  This is called \emph{entanglement concentration} if the initial state is pure \cite{BennettPRA96}, and \emph{entanglement distillation} if it is mixed \cite{BennettPRL96}.  For pure states, the appropriate measure of entanglement is given by the von Neumann entropy of the reduced density operator of either subsystem \cite{PopescuPRA97}.  A partially entangled pure state
\begin{equation}\label{eq:PartlyEntangledState}
\ket\psi = \sqrt{p}\ket{0_A0_B} + \sqrt{1-p}\ket{1_A1_B}
\end{equation}
has entanglement $H(p) = - (p\log p +(1-p)\log(1-p))$.  This means that if Alice and Bob collect $N$ pairs, and $N$ is large, then they can concentrate their entanglement into approximately $NH(p)$ EPR pairs.  Remarkably, this requires no communication; they can do it by independently performing local reversible computations \cite{BennettPRA96}.  However, existing protocols for entanglement concentration \cite{BennettPRA96,KayeJPhA01,MatsumotoPRA07} are block algorithms, that is, Alice and Bob must process all $N$ qubits together.  This approach has two drawbacks: \emph{lag} and \emph{memory}.  Alice and Bob get no EPR pairs until all $N$ input qubits have arrived, and they need $N$-qubit quantum computers to store and process all the input qubits.  The experimental state of the art -- roughly 10 qubits as of this writing -- cannot achieve the large block sizes required to approach optimality.  So let us explore what can be achieved with a small quantum information processor.

We could solve the lag and memory problems by breaking the input stream into blocks of length $N_0$, processing them one at a time.  But this also introduces error and/or inefficiency.  Not even a single \emph{perfect} EPR pair can be extracted with certainty from a block of finite length $N_0$.  If Alice and Bob are willing to settle for slightly distorted EPR pairs, then they can do much better.  They can extract $N_0H(p) - O(\sqrt{N_0})$ pairs, each of which has fidelity $1-e^{-O(N_0)}$ with a perfect EPR pair.  However, this protocol cannot approach the Shannon bound for fixed $N_0$; the $O(\sqrt{N_0})$ term represents wasted entanglement.  A better approach is to let each block yield a variable number of EPR pairs.  This achieves an average yield of up to $N_0H(p) - O(\log N_0)$ pairs per block, which still falls short of the Shannon bound for finite $N_0$.

In general, there might be a tension between two goals: achieving the Shannon bound for large $N$, and getting out perfect EPR pairs as quickly as possible for small and intermediate $N$.  In fact, these goals can both be achieved at the same time.  In this paper, we present a \emph{sequential} (a.k.a. instantaneous, streaming, or online) protocol that reads in partially entangled pairs one at a time and outputs perfect EPR pairs as they're generated.  
\begin{theorem}\label{ThmMain}
Let Alice and Bob share many copies of a bipartite pure state $\ket{\psi}$ with entanglement $E$.  There exists an entanglement concentration protocol that Alice and Bob run independently, in parallel, and sequentially on their sequences, which has the following properties.
\begin{enumerate}
\item After both parties have processed $N$ qubits, the expected yield is $NE - O(\log N)$ perfect EPR pairs -- i.e., the optimal rate is achieved.
\item This holds for every $N$, so the lag time is $O(\log N)$.
\item The algorithm works for all input states.
\item It uses only $O(\log N)$ qubits of memory.
\end{enumerate}
\end{theorem}
This protocol is fully reversible and coherent, involving no measurements.  As a result, the points in Theorem \ref{ThmMain} are not independent.  E.g., since the algorithm is reversible, it does not destroy any entanglement -- and since it runs in $O(\log N)$ memory, at least $NE - O(\log N)$ bits of entanglement must have been emitted at any time $N$.

The rest of this paper constitutes the proof of Theorem \ref{ThmMain}.  It is organized as follows.  In Section \ref{SecCompression} we discuss data compression and show that quantum variable length compression codes are not suitable for entanglement concentration.  In section \ref{SecRandomness} we turn to classical randomness extraction, and discuss Elias's optimal block extractor.  In section \ref{SecStreaming} we construct a streaming version of Elias's randomness extraction protocol, and show that it can be used for entanglement concentration when the Schmidt basis is known.  In Section \ref{SecSchur} we build a fully universal protocol by combining our extraction protocol with the quantum Schur transform.


\section{Data compression:  why it doesn't work}\label{SecCompression}

There is a deep link between entanglement concentration and quantum data compression.  Given the state in Eq. \ref{eq:PartlyEntangledState}, Alice and Bob each describe their $n$th input qubit by a density matrix
\begin{equation}
\rho_A = \rho_B = p\proj{0}+(1-p)\proj{1},
\end{equation}
with entropy $H(p)\leq1$, all of which is due to entanglement with the other party.  Concentrating the entanglement contained in $N$ input qubits into $M$ EPR pairs, for which the parties' reduced states are
\begin{equation}
\rho'_A = \rho'_B = \frac12\left(\proj{0}+\proj{1}\right),
\end{equation}
means compressing the entropy of $N$ copies of $\rho$ into $M$ maximally mixed states.  Done reversibly, this is data compression.  Indeed, the original entanglement concentration protocol of Bennett et. al. \cite{BennettPRA96} is essentially a block compression algorithm, followed by a measurement of Hamming weight.

Seeking a sequential protocol for entanglement concentration, we might therefore turn to sequential data compression protocols.  Some of the oldest and best-known methods of classical data compression are of this type.  Variable-length protocols such as Huffman coding and arithmetic coding replace each input symbol with a codeword whose length depends on the symbol's probability.  Quantum algorithms for Huffman coding and arithmetic coding exist \cite{BraunsteinIEEE00,ChuangIEEE00} (the Chuang-Modha algorithm for arithmetic coding is actually a block protocol, but there's no fundamental obstacle to sequential quantum arithmetic coding).  However, the total length of the transmission is entangled with the messages being sent.  So, although the encoder can compress sequentially, the decoder must wait until the end of the transmission to start decoding.

For this reason, variable-length compression does not accomplish entanglement concentration.  Even under optimal circumstances (i.e., where a Huffman code with block-length 1 achieves the Shannon bound for compression), Alice and Bob's output qubits are not perfect EPR pairs.  Although each party's $n$th output qubit is indeed maximally mixed (as it should be, if it's to be half of an EPR pair), it is correlated with subsequent output qubits, e.g. the $(n+1)$th qubit.  This correlation decoheres the EPR pairs.

To see a simple example of this, consider a 4-dimensional input Hilbert space spanned by $\{\ket{a},\ket{b},\ket{c},\ket{d}\}$, and a source that emits
\begin{eqnarray}
\ket{\psi}_{\mathrm{in}} = \frac{1}{\sqrt2}\ket{aa} + \frac{1}{\sqrt4}\ket{bb} + \frac{1}{\sqrt8}\ket{cc} + \frac{1}{\sqrt8}\ket{dd} &&\\
\Longrightarrow \rho_A = \rho_B = \frac12\proj{a} + \frac14\proj{b} + \frac18\proj{c} + \frac18\proj{d}&&
\end{eqnarray}
This distribution, with entropy $H=1.75$ bits, can be compressed perfectly into qubits by the following Huffman code:
\begin{eqnarray}
\ket{a} &\to& \ket{0} \nonumber \\
\ket{b} &\to& \ket{10} \nonumber \\
\ket{c} &\to& \ket{110} \nonumber \\
\ket{d} &\to& \ket{111}. 
\end{eqnarray}
If Alice and Bob each apply this protocol to their input streams, the \emph{first} partially-entangled pair is transformed to
\begin{eqnarray}
\ket{\psi}_{\mathrm{out}} &= &\ \ \frac{1}{\sqrt2}\ket{0}_A\ket{0}_B + \frac{1}{2}\ket{10}_A\ket{10}_B  \vspace{0.1in} \\ &&+\frac{1}{2\sqrt2}\ket{110}_A\ket{110}_B + \frac{1}{2\sqrt2}\ket{111}_A\ket{111}_B.\nonumber
\end{eqnarray}
Consider the reduced state of Alice's and Bob's first output bits, obtained by tracing out the 2nd and 3rd bits (the string is implicitly zero-padded, so unspecified bits are in $\ket{0}$).  In the basis $\{\ket{00},\ket{01},\ket{10},\ket{11}\}$, it is
\begin{equation}
\rho_{\mathrm{out}} = \pmat{\frac12 & 0 & 0 & \frac{1}{2\sqrt2} \\
0 & 0 & 0 & 0 \\
0 & 0 & 0 & 0 \\
\frac{1}{2\sqrt2} & 0 & 0 & \frac12},
\end{equation}
This state's fidelity with an EPR state is only $\frac{1+\sqrt2}{2\sqrt2} \approx 0.85$.  Furthermore (and this is important!), since this protocol is sequential, it will never go back and change the first bit.  Nothing that Alice and Bob do to subsequent output bits can enhance the entanglement of their first pair; it will always be defective.

This failure reflects an inherent property of variable-length codes:  each output symbol is correlated with the length of the entire output (see \cite{BraunsteinIEEE00} for the first mention of this issue, but in a different context).  The correlation is indirect, for both the individual output symbols and their overall length are determined by the input symbols.  For a rather extreme example, recall that the Huffman code given above maps $\ket{a}\to\ket{0}$ and $\ket{d}\to\ket{111}$.  If the output string contains high proportion of $\ket{0}$ qubits, then the input string must have contained a lot of $\ket{a}$ symbols, and therefore the output string will be relatively short.  A high proportion of $\ket{1}$ qubits, on the other hand, means that the input contained a lot of $\ket{d}$ symbols, and so the output is relatively long.  This correlation is enough to decohere each individual output EPR pair.

\section{Extracting randomness}\label{SecRandomness}

This failed experiment in using standard data compression demonstrates a key point:  in sequential concentration, Alice's $n$th output qubit must not be correlated with \emph{anything} except Bob's $n$th output qubit (and vice-versa), from the moment it is written down.  No subsequent actions by the concentrator can fix a defective output.  The first step in a protocol that emits a stream of perfect EPR pairs is to generate just \emph{one} perfect EPR pair.  This cannot be done deterministically with a finite number of input qubits, but it can be done conditionally -- i.e., \emph{if} a pair is generated, then it is perfect.

If Alice and Bob know their shared state, then extracting a perfect EPR pair is closely related to a classical problem:  ``How do we extract a perfect \emph{independent} random bit from a stream of biased, i.i.d., random bits?''  It is critical that each extracted bit be independent of \emph{everything}, including other random bits and the processor's memory.

\subsection{Von Neumann's protocol}

Von Neumann addressed this problem in 1951 \cite{VonNeumann51}.  He proposed sampling the biased bits two at a time.  The odd-parity sequences ``01'' and ``10'' have equal probability, so if the first two bits have odd parity, Von Neumann reports the first bit.  If we draw two bits with even parity (``00'' or ``11''), we discard them and draw another pair.

Each time a pair is drawn, the Von Neumann scheme emits a random bit with probability $2p_0p_1$, and fails with probability $1-2p_0p_1$.  The number of input bit pairs required to get a single random bit is exponentially distributed,
\begin{equation}
\mathrm{Pr}(n) = 2p_0p_1(1-2p_0p_1)^{n-1},
\end{equation}
and the expected waiting time for the first random bit is
\begin{equation}
\expect{n} = \frac{1}{p_0p_1}.
\end{equation}
Since the protocol is completely Markovian, the rate at which randomness is extracted is
\begin{equation}
R \equiv \frac{\diff N_{\mathrm{rbits}}}{\diff n} = p_0p_1.
\end{equation}
This is quite a bit less than the theoretical upper bound, $R_{\mathrm{max}} = H(p_0)$, because Von Neumann's protocol wastes a lot of entropy. However, it is a sequential protocol, and it can be used for entanglement concentration.  Alice and Bob each run the following algorithm:
\begin{enumerate}
\item Draw two qubits $q_1,q_2$ from the input.
\item Perform a CNOT (in the Schmidt basis) from $q_1\to q_2$.  This stores the parity, $q_1\oplus q_2$, in $q_2$.
\item Conditional on $q_2=\proj{1}$, swap $q_1$ with the ``output register", and halt; otherwise draw two new qubits and repeat.
\end{enumerate}
Thus, if Alice and Bob share multiple copies of the entangled state
\begin{equation}
\ket{\psi} = \sqrt{p_0} \ket{0_A 0_B} + \sqrt{p_1} \ket{1_A 1_B},
\end{equation}
the first two copies are tranformed as follows:
\begin{eqnarray}
\ket{\psi \psi} &=& p_0 \ket{00_A} \ket{00_B} + p_1 \ket{11_A} \ket{11_B} \nonumber \\
&& + \sqrt{p_0 p_1} \left( \ket{01_A} \ket{01_B} + \ket{10_A} \ket{10_B} \right) \nonumber \\
&\Rightarrow& \left( p_0 \ket{0_A 0_B} + p_1 \ket{1_A 1_B} \right) \ket{0_A 0_B} \nonumber \\
&& + \sqrt{p_0 p_1} \left(\ket{0_A 0_B} + \ket{1_A 1_B} \right) \ket{1_A 1_B}.
\end{eqnarray}
Conditional on Alice and Bob's second qubits each being in state $\ket{1}$, their first qubits now form a perfect EPR pair, $\ket{\Phi^+}$.  Otherwise, their joint state is given by 
\begin{equation}
\ket{\psi}_{\mathrm{fail}} = \frac{p_0}{\sqrt{1-2 p_0 p_1}} \ket{0_A 0_B} + \frac{p_1}{\sqrt{1-2 p_0 p_1}} \ket{1_A 1_B},
\end{equation}
and they each read another two qubits and repeat.  Note that there is substantial entanglement left in $\ket{\psi}_{\mathrm{fail}}$.  In Von Neumann's protocol, this entanglement is wasted, and we will get a better protocol by recycling it.

The protocol given above continues to draw pairs until it succeeds, at which point it deposits an EPR pair into Alice and Bob's first qubits and halts.  Running this coherently and in parallel on $2N$ copies of $\ket{\psi}$ gives
\begin{eqnarray}
\ket{\psi \psi}^{\otimes N} &\to& \sqrt{2 p_0 p_1} \ket{\Phi^+} \left( \sum_{k=0}^{N-1} (1 - 2 p_0 p_1)^{k/2} \right. \nonumber \\
&& \left. ( \ket{0_A 0_B} \ket{\psi}_{\mathrm{fail}} )^{\otimes k} \ket{1_A 1_B} \ket{\psi \psi}^{\otimes(N-k-1)}  \right)\nonumber \\
&& + (1-2p_0 p_1)^{N/2} (\ket{\psi}_{\mathrm{fail}} \ket{0_A 0_B})^{\otimes N}.
\end{eqnarray}
The amplitude for not halting decreases exponentially with $N$, so for moderately large $N$ we can be nearly certain that an EPR pair has been deposited.

This quantum Von Neumann protocol uses an indeterminate and unbounded number ($2k+2$) of input pairs to produce a single perfect EPR pair.  If we want a perfect EPR pair with certainty, then $N_{\mathrm{input}}$ \emph{must} be unbounded.  A finite-sized block of partially entangled states does not generally contain even a single perfect EPR pair.  Fortunately, $k$ is exponentially distributed, so we can get an extraordinarily good EPR pair by terminating the algorithm at relatively small $k$.

The algorithm can be iterated, without any modification, to extract a stream of EPR pairs.  This is ``on-demand'' mode:  the user requests exactly 1 (or $n$) EPR pairs, and the protocol reads as many input pairs as are needed.  If we wait for a near-perfect EPR pair, then a lot of time is wasted.  The algorithm probably (i.e., with large amplitude) halts at small $k$, and yet achieving near certainty mandates waiting for longer (but low-amplitude) computational paths to terminate.

Alternatively, we could replace the output register with an output tape, and replace the ``halt'' instruction with ``push $q_1$ onto the tape and shift it by one qubit.''  Now, the algorithm never terminates unless it runs out of inputs.  As soon as it produces one EPR pair, it starts working on the next.  In this fully streaming mode, the length of the output tape is always indeterminate, but $k$ (the number of bits read so far) can be well-defined (e.g., if the algorithm terminates).

The protocols we will design in subsequent sections can be run in either mode.  We will typically focus on the fully streaming mode, where the output tape's length is indeterminate, because it is compatible with a bounded input tape.  In the quantum Von Neumann protocol, this mode is relatively unproblematic.  To get an EPR pair, the user pops one off the end of the tape (without learning how long the output tape is).  A problem occurs only if the user finds \emph{no} available pairs, which implies that the output tape is empty.  

This is \emph{not} true for other protocols, which recycle entanglement in order to achieve much higher efficiency.  This recycling requires a coherent superposition of many output tape lengths. Disrupting this superposition (by issuing a failed request for an EPR pair that is, with some amplitude, not available) will reduce efficiency.  So in these protocols, the first few squares of the output tape must be regarded as a sort of incubator -- a region where EPR pairs are almost certainly available, but should nonetheless not be used.  Running the protocol in on-demand mode avoids this problem entirely (but requires an unbounded stream of inputs).

\subsection{Achieving the Shannon bound:  Elias's protocol}

Von Neumann's protocol wastes at least 75\% of the entropy in the input bits, and the corresponding entanglement concentration protocol wastes an equal amount of entanglement.  Block protocols, in contrast, can extract randomness or entanglement with asymptotically perfect efficiency -- i.e., at a rate given by the entropy of the source, as $N\to\infty$.  We will now develop a sequential protocol that achieves the entropic bound as $N\to\infty$.  In fact, our protocol is a sequential implementation of the optimal block protocol, and extracts at most 2 ebits less than it.

Quite a few papers have followed up on Von Neumann's work, generalizing and improving it.  Early work focused on the extraction of a single random bit, and sought to minimize the expected number of input bits.  Hoeffding and Simons \cite{Hoeffding70} represented algorithms as random walks on the lattice of non-negative integer points in the plane, $\{n_0,n_1\}$.  Stout and Warren \cite{Stout84} represented algorithms more generally as walks on binary trees.  Other authors (notably Samuelson \cite{Samuelson68} and Elias \cite{Elias72}) showed how to extract random bits from $k$th-order Markov processes, a particular kind of non-i.i.d.~source.  A flood of more recent work (beginning with Trevisan's seminal paper in 1998) has generalized the notion of \emph{extractors} to extremely general non-i.i.d.~sources; this level of generalization, however, is not relevant to our task.

Each of these single-bit extraction protocols can be repeated (like Von Neumann's) to yield a stream of random bits (or EPR pairs, in the context of entanglement concentration).  Such protocols never approach the Shannon bound, since any residual entropy/entanglement in the used input bits is wasted (Hoeffding and Simons \cite{Hoeffding70} proved an upper bound of $R=1/3$ on the rate, and demonstrated an algorithm that achieves $R\approx0.323$ as $p\to\frac12$).  An efficient protocol has to somehow recycle this entropy.

Elias seems to have been both the first and the last to suggest an asymptotically efficient block protocol \cite{Elias72}.  Elias's protocol, which is essentially unimprovable, uses the fact that every $N$-bit string containing $T$ ``1'' bits has probability
\begin{equation}
\mathrm{Pr}(N,T) = p_0^{N-T}p_1^T
\end{equation}
The set of all such strings is a \emph{type class}, containing exactly $\binom{N}{T}$ strings with the same probability.  If we draw an $N$-bit string, then conditional on the type being $T$, the index $\alpha \in\left[1\ldots\binom{N}{T}\right]$ of the \emph{particular} string drawn is a uniformly random variable with $\binom{N}{T}$ possible values.  If $\binom{N}{T}$ happens to equal $2^L$, then by writing this index down in binary, we immediately get $L$ perfectly random bits.  Otherwise, we use the binary representation of $\binom{N}{T}$ to expand it as a sum of powers of 2,
\begin{equation}
\binom{N}{T} = 2^{L_1} + 2^{L_2} + \ldots + 2^{L_n},
\end{equation}
and divide the interval $\mathcal{I} = \left[1\ldots\binom{N}{T}\right]$ into bins
\begin{eqnarray*}
\mathcal{I}_{L_1} &=& [1\ldots 2^{L_1}]\\
\mathcal{I}_{L_2} &=& [2^{L_1}+1\ldots 2^{L_1}+2^{L_2}]\\
&\hdots&\\
\mathcal{I}_{L_k} &=& \left[\left(\sum_{j=1}^{k-1}{2^{L_k}}\right)+1\ldots \left(\sum_{j=1}^{k-1}{2^{L_k}}\right)+2^{L_k}\right]\\
&\hdots&
\end{eqnarray*}
If the index $\alpha$ lies in the interval $\mathcal{I}_{L_k}$, we output $\alpha_{L_k} = \alpha - \left(\sum_{j=1}^{k-1}{2^{L_k}}\right)$ as a $L_k$-bit string.

\begin{theorem}\label{thm:EliasExtraction}
On average, Elias's protocol extracts at least $NH(p) - \log_2(N+1)-2$ bits of entropy from $N$ input bits, so as $N\to\infty$, it achieves the Shannon bound.
\end{theorem}

\begin{proof}
To prove this, we let $s$ be an $N$-bit string, and observe that $s$ is equivalent to a pair of indices $(T,\alpha)$, where $T$ is its type and $\alpha$ the index of $s$ within type T.  Furthermore, $\alpha$ is equivalent to a pair $(L,\alpha_L)$, where $\mathcal{I}_L$ is the bin of $\binom{N}{T}$ containing $\alpha$, and $\alpha_L$ is its $L$-bit index within $\mathcal{I}_L$.  Since we output the entirety of $\alpha_L$, the extracted entropy is $H(\alpha_L|T,L)$.  Now, since $s \sim (T,L,\alpha_L)$, we apply the chain rule for conditional entropy,
\begin{equation}
H(Y|X) = H(Y,X) - H(X),
\end{equation}
to obtain
\begin{eqnarray*}
H(\alpha_L|T,L) &=& H(T,L,\alpha_L) - H(T,L) \\
&=& H(s) - \left[H(T)+H(L|T)\right].
\end{eqnarray*}
The input distribution has exactly $NH(p)$ bits of entropy, so $H(s) = NH(p)$.  Since there are only $N+1$ types, $H(T)\leq\log_2(N+1)$ (actually, $T$ is binomially distributed, so $H(T) = \frac12\log_2(N)+O(1)$).  To calculate $H(L|T)$, recall that
\begin{equation}
\binom{N}{T} = \sum_k{2^{L_k}},
\end{equation}
so $L$ takes values $\{L_k\}$ with probability
\begin{equation}
P(L_k) = \binom{N}{T}^{-1}2^{L_k},
\end{equation}
and $H(L|T)$ is just the entropy of this distribution.  Now, we can place an upper bound of $2$ bits on $H(L|T)$ by the following argument:

Let $n$ be an integer, with a binary expansion
\begin{equation}
n = \sum_{k}{2^{L_k}}
\end{equation}
where $L_0>L_1>\ldots L_K$.  This defines a probability distribution over $L$,
\begin{equation}
\mathrm{Pr}(L_k)\equiv \mathrm{Pr}(L=L_k) = \frac{2^{L_k}}{n}.
\end{equation}
Now, since $L_0>L_1>\ldots L_K$, then it's easy to see that $\mathrm{Pr}(L_0)>\frac12$, and that $\mathrm{Pr}(L_1)>\frac12(1-\mathrm{Pr}(L_0))$, and in general that $\mathrm{Pr}(L_k)>\frac12\mathrm{Pr}(L \leq L_k)$.  Thus $\mathrm{Pr}(L)$ majorizes the infinite exponential distribution given by
\begin{equation}
\mathrm{Pr}(l) = 2^{-l}\mathrm{\ :\ }l=1\ldots\infty
\end{equation}
whose entropy is exactly 2 bits.  Since entropy is convex, $H(L|T)\leq H(l)=2$.
\end{proof}

Like Von Neumann's protocol, Elias's protocol, if performed coherently, gives an entanglement concentration protocol.  The original block concentration protocol~\cite{BennettPRA96} uses a decomposition very similar to Elias's, while subsequent work by Kaye and Mosca uses exactly this decomposition~\cite{KayeJPhA01}.  Whereas Von Neumann's protocol yields either 0 or 1 EPR pairs, and can be repeated conditional on failure to yield exactly 1 pair, Elias's protocol yields a variable, binomially distributed number of EPR pairs.

\begin{theorem}\label{thm:EliasConcentration}
Elias's protocol, performed coherently on $N$ copies of the bipartite state $\ket\psi$, has an average yield of at least $NH(\rho) - \log_2(N+1)-2$ EPR pairs, where $\rho = \Tr_B\proj{\psi}$.
\end{theorem}

\begin{proof}
Alice and Bob begin with the state
\begin{eqnarray*}
\ket\psi^{\otimes N} &=& \left(\sqrt{p}\ket{0_A0_B} + \sqrt{1-p}\ket{1_A1_B}\right)^{\otimes N} \\
&=& \sum_{s\in\{0,1\}^{N}}{\sqrt{\mathrm{Pr}(s)}\ket{s_A s_B}}
\end{eqnarray*}
where the probability $\mathrm{Pr}(s)$ of a string $s\in\{0,1\}^{N}$ containing $T(s)$ ``1''s is
$$\mathrm{Pr}(s) = p^{N-T(s)}(1-p)^{T(s)}.$$
Each type class $\cT_T$, labeled by its Hamming weight $T$, defines a type subspace spanned by $\ket{s_As_B}$ for all $s$ in the type class.  We can rewrite the joint state as a sum over type subspaces,
$$\ket\psi^{\otimes N} = \sum_T{ \sqrt{\mathrm{Pr}(T)} \frac{\sum_{s\in\cT_T}{\ket{s_A s_B}}} {\sqrt{\binom{N}{T}} } },$$
where $\mathrm{Pr}(T) = \binom{N}{T}p^{N-T}(1-p)^T$.  So if Alice and Bob both measure $T$, then they both obtain the \emph{same} value $\hat{T}$, which is distributed according to $\mathrm{Pr}(\hat{T})$.  Conditional on this measurement, they have
$$\ket\psi_{\hat{T}} = \frac{\sum_{s\in\cT_{\hat{T}}}{\ket{s_A s_B}}}{\sqrt{\binom{N}{\hat{T}}}},$$
which is a maximally entangled state of dimension $\binom{N}{\hat{T}}$.  They now divide the strings of type $\hat{T}$ into bins $\cL_L$ of size $2^L$.  The specific binning is entirely arbitrary, as long as Alice and Bob use the same one.  Alice and Bob's state is
$$\ket\psi_{\hat{T}} = \sum_L{\sqrt{\frac{2^L}{\binom{N}{\hat{T}}}} \frac{\sum_{s\in\cL_L}\ket{s_As_B}}{\sqrt{2^L}}}.$$
As with $T$ above, Alice and Bob can measure $L$ and be assured of getting the same answer $\hat{L}$.  Conditional on $\hat{L}$, they have
$$\ket\psi_{\hat{T},\hat{L}} = \frac{\sum_{s\in\cL_{\hat{L}}}{\ket{s_A s_B}}}{\sqrt{2^L}} = \left(\frac{\ket{0_A0_B} + \ket{1_A1_B}}{\sqrt2}\right)^{\otimes \hat{L}},$$
so Alice and Bob now share $\hat{L}$ EPR pairs (although they are still distributed over $N$ physical qubits).  The joint distribution $\mathrm{Pr}(T,L)$ is identical by inspection to the one in Theorem \ref{thm:EliasExtraction}, so expected yield is identical.
\end{proof}

\section{Streaming extraction}\label{SecStreaming}

Elias's protocol is a block algorithm; it operates on all $N$ qubits at once.  There have been relatively few attempts to design efficient sequential extractors.  Several authors (including Elias) have observed that single-shot protocols such as the Von Neumann or the Hoeffding-Simons protocol can be repeated indefinitely, but that they are far from optimal.  Elias suggested a quasi-sequential application of his protocol:  apply it to the first 2 input bits, then the next 4, then the next 6, etc, etc.  This is both strictly suboptimal for any $N$ (though it does approach the Shannon bound as $N\to\infty$), and memory-intensive as $N\to\infty$ (since the blocklength grows as $\sqrt{N}$).  Peres \cite{Peres92} showed how to iterate Von Neumann's protocol, recycling the entropy in bits that have already been used, but his protocol is not actually sequential.  Visweswariah et al.~\cite{VisweswariahIEEE98} suggested the use of variable-length source codes as extractors, but Hayashi \cite{HayashiIEEE08} subsequently pointed out that the output bits are not quite randomly distributed (see also our discussion above of the problems this raises for entanglement concentration).

Our first goal is to construct a sequential extractor that achieves the Shannon bound (in fact, a streaming implementation of Elias's protocol).  That is, we wish to construct an algorithm that reads bits one at a time, performing some processing and outputting random bits as they are produced, before reading the next bit.  Further, when $N$ bits have been read, for any given $N$, our protocol extracts the same amount of randomness as Elias's block protocol.  We will first assume we have applied Elias's protocol to a block of size $N-1$, and investigate how to extract the extra randomness produced by adding one more input bit.

\subsection{Serializing Elias's protocol}

We have seen that Elias's protocol represents an $N$-bit input string $s$ as $(T,L,\alpha_L)$, where $T$ describes the type, $L$ represents the bin $\mathcal{I}_L$ within the type, and $\alpha_L$ the $L$-bit index within $\mathcal{I}_L$.    The index $\alpha_L$ consists of $L$ perfectly random bits, and forms the output of the protocol.  A particular implementation of the protocol provides a particular mapping from the strings within a given type to the index pair $(L,\alpha_L)$.  In order to construct a streaming implementation, we first note that a mapping for $N$-bit strings may be constructed in a convenient way from the mapping for $(N-1)$-bit strings.  Suppose that the $N-1$ input bits $(b_1\ldots b_{N-1})$ have already been transformed into $(N-1,T_0,L_0,\alpha_{L_0})$ by Elias's protocol, and the $L_0$ random bits represented by $\alpha_{L_0}$ have been emitted.  We want to add one more bit $b_N$, updating the transformation as $(N-1,T_0,L_0,\alpha_{L_0})\to(N,T,L,\alpha_{L})$.  Since a streaming protocol acts on strings of different lengths, we also keep track of $N$, the number of bits read so far.  We will now describe this procedure in more detail.


Recall that each of the $\binom{N}{T}$ strings of type $(N,T)$ can be obtained \emph{either} by adding a ``0'' to one of the strings of type $(N-1,T)$, or by adding a ``1'' to one of the strings of type $(N-1,T-1)$.
The strings of type $(N-1,T)$ have been sorted into bins such that bin $L$, if present, contains $2^L$ strings, and no two bins have the same size.  Similarly for the strings of type $(N-1, T-1)$.  When we find ourselves in a bin $L$, that means we have already outputted $L$ random bits.  Except for the value of those random bits, we treat all strings in the bin identically.  We don't want any two bins to be the same size, because in that case, we could combine the two bins into a single bin of twice the size, allowing us to output an additional random bit.  When we add an extra input bit and find ourselves now in type $(N,T)$, we wish to see if we can merge any bins, thus producing additional random output bits.

The sizes of the types satisfy a recursion rule:
\begin{equation}
\binom{N}{T} = \binom{N-1}{T} + \binom{N-1}{T-1}.
\end{equation}
Upon reading a new bit we update $N$ and $T$ to correspond to the new number of bits read and the new type.  We also wish to use the mapping into bins and indices $(L,\alpha_L)$ for $(N-1)$-bit strings to define one for $N$-bit strings.  Denoting
\begin{eqnarray}
\binom{N-1}{T} &=& \sum_j 2^{L_j}, \\
\binom{N-1}{T-1} &=& \sum_k 2^{L_k^{\prime}}, \\
\binom{N}{T} &=& \sum_i 2^{L_i^{\prime \prime}}
\end{eqnarray}
we obtain
\begin{equation}
\sum_i 2^{L_i^{\prime \prime}} = \sum_{j} 2^{L_j} + \sum_k 2^{L_k^{\prime}}.
\end{equation}
This is simply binary addition, and the rules of binary addition also tell us how to update the bins $L$ and indices $\alpha_L$.  If both the types $(N-1,T-1)$ and $(N-1,T)$ have a bin of size $2^L$, we can merge them, outputting a new random bit.  This gives us a new ``carry bin'' of size $2^{L+1}$, corresponding to the carry bit.  Perhaps we can merge this bin as well with another bin, producing another random bit and a new carry bin, and so on.

To construct a streaming implementation of Elias' protocol, we simply read bits one at a time, performing the above processing at each step.  It is easily verified that at $N=2$ the above performs von Neumann's protocol, while for $N>2$, by induction, following these rules gives an implementation of Elias' protocol for each $N$.  The rules for what to do upon reading the $N$-th bit are defined by the triplet $(N-1,T,L)$, along with the bit just read.  In particular, they do not depend on the index $\alpha_L$ that identifies a particular string.  So, since the $L$ bits of $\alpha_L$ are not needed to process subsequent bits, they can be ejected as soon as they are produced.  

For every input string that causes the $n$th output bit to be ``0'', there is a matching string that (a) produces \emph{exactly} the same memory state, (b) produces \emph{exactly} the same output bits except for the $n$th one, and (c) yields a ``1'' for the $n$th output bit.  This guarantees that the output bits are unbiased \emph{and} uncorrelated with the memory.  The memory state is completely specified by the three integers $(N,T,L)$, so the processor's memory need grow only as $O(\log(N))$.

\subsection{Implementation}

This implementation may be conveniently represented by the lattice shown (up to $N=5$) below.  Each possible input string corresponds to a different path through the lattice, and red dots indicate the fusion of two paths into a single node by outputting an unbiased bit.

\begin{center}
\includegraphics[width=\HPW]{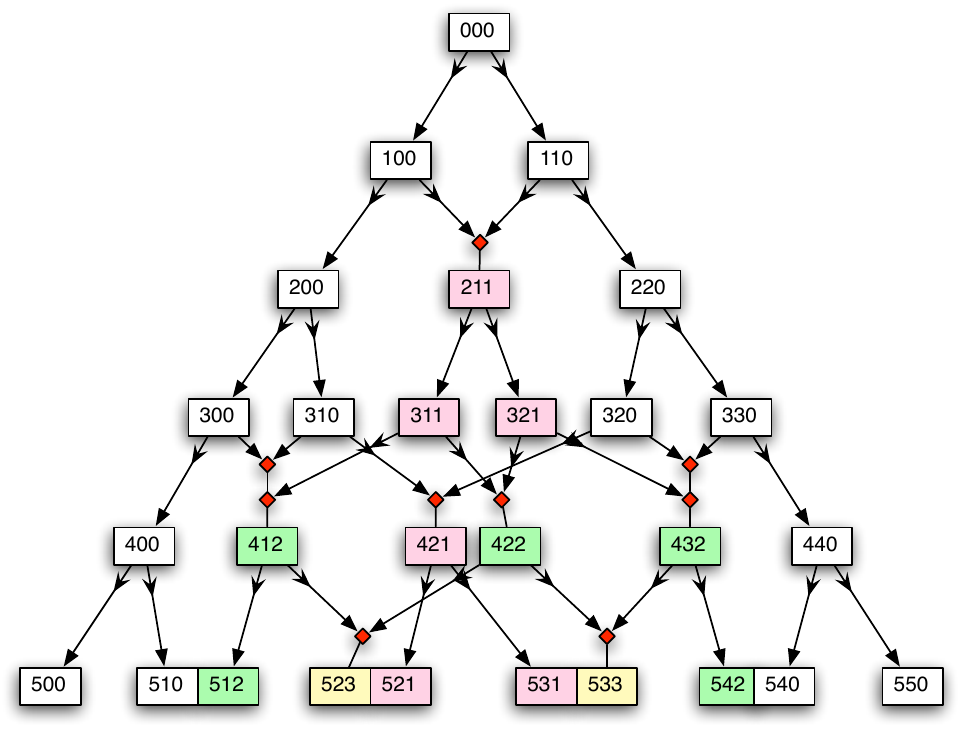}
\end{center}

The nodes are labeled by three integers:

\begin{itemize}

\item $N$, the number of input bits read so far,

\item $T$, the Hamming weight of the input string,

\item $L$, the number of random bits output so far.

\end{itemize}
Note that this lattice is simply the lattice corresponding to Pascal's triangle, that is with each node representing a different type class, but with the nodes $(N,T)$ subdivided into $\{(N,T,L)\}$ for each value of $L$ in the binary expansion of $\binom{N}{T}$.  Since $(N,T,L)$ represents a collection of $2^L$ strings with the same probability, $L$ will represent the number of random bits output so far.  The procedure in the previous section tells us how to traverse the lattice.  We can formalize the lattice traversal with the following rules.

\begin{protocol} \label{Prot1}
This protocol runs on a machine with three integer memory registers labeled $N$, $T$, and $L$.  Define $\binom{N}{T}_L$ to be the $L$th bit of $\binom{N}{T}$.











{\tt
\noindent\ 1 \textbf{WHILE}( input stream not empty ) \textbf{DO} 

\vspace{0.035in}\noindent\ 2 \{

\vspace{0.035in}\noindent\ 3\hspace{0.2in}Read a bit $b$ from the input stream.

\vspace{0.035in}\noindent\ 4\hspace{0.2in}Update $N\to N+1$ and $T\to T+b$.

\vspace{0.035in}\noindent\ 5\hspace{0.2in}\textbf{IF}( $\binom{N}{T}_L=0$ or $\binom{N-1}{T-1+b}_L = 1$ )

\vspace{0.035in}\noindent\ 6\hspace{0.2in}\{

\vspace{0.035in}\noindent\ 7\hspace{0.4in}output $b$ and set $L\to L+1$.

\vspace{0.035in}\noindent\ 8\hspace{0.4in}\textbf{WHILE}( $\binom{N-1}{T}_L \neq \binom{N-1}{T-1}_L$ )

\vspace{0.035in}\noindent\ 9 \hspace{0.6in}output $\binom{N-1}{T}_L$ and set $L\to L+1$.

\vspace{0.035in}\noindent10 \hspace{0.2in}\}

\vspace{0.035in}\noindent11 \}
}

\end{protocol}
\noindent\textbf{Discussion:}  When stated concisely, the protocol is a bit cryptic, so here is an explanation of how (and why) it works.  First, note that the protocol as given runs in ``fully streaming'' mode -- i.e., it continues to read and write bits indefinitely.  To make it run in ``on-demand'' mode, we change each instance of ``output $x$'' to ``output $x$ and then pause.''  It's \emph{not} sufficient to pause before line 3, because (significantly) there is not a 1:1 correspondence between reading and outputting bits.

The basic idea here is to read bits until the algorithm arrives at an internal state $(N,T,L)$ that could have been reached via two different paths with equal probability.  One path comes from $(N-1,T-1,L)$ by reading $b=1$, while the other comes from $(N-1,T,L)$ by reading $b=0$.  Since $b$ identifies the path, and the two paths are equally probable, $b$ is perfectly random.  So the algorithm spits it out.

This simple picture gets complicated because of \emph{carrying}.  The nodes are in 1:1 correspondence with bits of binomial coefficients.  The existence of two paths coming from $(N-1,T-1,L)$ and $(N-1,T,L)$ means that the $L$th bits of $\binom{N-1}{T-1}$ and $\binom{N-1}{T}$ are both 1.  Adding them produces a carry bit in column $L+1$.  Fusing the corresponding paths produces a ``carry path'' corresponding to node $(N,T,L+1)$.  If that node could have been reached in another way (from either or both of $(N-1,T-1,L+1)$ or $(N-1,T,L+1)$), then we need to fuse some more paths.

The algorithm begins by reading a bit and updating its internal state.  Line 5 checks to see whether the resulting internal state could have been reached in at least two ways.  If not, then no perfectly random bit is available, and it reads another bit.  How is this check performed?  Each node can be reached via one, two, or three paths.  $(N,T,L)$ has exactly one path leading into it if $\binom{N}{T}_L=1$ (meaning there's either one or three paths in), \emph{and} $\binom{N-1}{T-1+b}_L=0$ (meaning there's only one non-carry path, and thus no more than two paths in total).  So if either of these is false, then there are two or three paths leading in.   

This could actually happen in three different ways.  There could be:
\begin{enumerate}
\item Two non-carry paths,
\item One non-carry path and one carry path,
\item Three paths.
\end{enumerate}
By design, the algorithm outputs $b$ in all three sub-cases (Line 7).  It also updates $L\to L+1$, then proceeds to deal with the resulting carry path into $(N,T,L+1)$.  

There is some freedom in how to deal with carry paths, but the rule embodied by Line 7 eliminates all of it.  Recall that output bits have to be uniformly random.  If a node has two non-carry paths (and no carry path) leading in, then we're in good shape -- those paths have identical probability, but different values of $b$, so the output bit is random.  If there is only one non-carry path, and it entailed reading in $b=\hat{b}$, then the corresponding carry path \emph{must} output $1-\hat{b}$.  Finally, if there are three paths, then only two of them can fuse.  By outputting $b$, we are choosing to fuse the two non-carry paths -- so the corresponding carry path must \emph{not} output anything.

Lines 8-9 ensure that the carry path is managed correctly.  When the algorithm finds itself in $(N,T,L)$ via a carry path, there are two cases:
\begin{enumerate}
\item If there are either 0 or 2 non-carry paths into $(N,T,L)$, then there's no path to fuse with, so the algorithm should output nothing and read another bit.  This is true if the $L$th bits of $\binom{N-1}{T}$ and $\binom{N-1}{T-1}$ are the same (both zero means 0 non-carry paths, while both 1 means 2 non-carry paths).
\item If those bits are different, then there is exactly 1 non-carry path into $(N,T,L)$.  If that path came from $(N-1,T,L)$, then its last input bit would have been $\hat{b}=0$.  To fuse with that path, our algorithm should output a 1.  Otherwise, the non-carry path came from $(N-1,T-1,L)$, its last input bit would have been $\hat{b}=1$, and the algorithm should output a 0.  In either case, $\binom{N-1}{T}_L$ gives the correct output.  Outputting a bit yields another carry path, so the \texttt{WHILE} statement ensures that we loop around to line 8 and deal with it in turn.
\end{enumerate}

We are going to use this algorithm as an entanglement concentration protocol, so it has to be completely reversible.  In the description above, inputs and outputs are asynchronous -- and the algorithm generally has to read bits at a higher rate than it can output EPR pairs.  If these bits are physical systems (e.g., qubits), where are they going?  

To clarify this, we assume that the machine has access to three I/O bitstreams or ``tapes''.  The \emph{input tape} is read-only, the \emph{output tape} is write-only, and the \emph{purity tape} is a read/write stack that functions as a reservoir of clean ``0'' bits.  Now the protocol is explicitly reversible: of $N$ input bits, $n$ will be pushed onto the output tape, and $N-n$ will be pushed onto the purity tape.  However, upon reading in a bit $b$, the protocol may pop one or more bits off the purity tape, write random bits onto them, and push them onto the output tape (line 80).  On the other hand, it may also erase $b$ (i.e., reversibly set it to ``0'') and push it onto the purity tape (line 40).  Lines 50 and 70 do not require any action on the purity tape.

We can summarize this construction as a theorem, whose proof is the preceding analysis:
\begin{theorem}
Protocol \ref{Prot1}, applied to a series of $N$ bits, implements Elias's protocol for optimal randomness extraction.
\end{theorem}

\subsection{Performance}

The algorithm described above is a sequential protocol for extracting perfectly random bits.  But how well does it work?

We begin by noting that our algorithm is sequential, but not instantaneous.  A truly instantaneous protocol (like Huffman coding) is Markovian.  Its action on a given input symbol does not depend on previous symbols, so it requires no memory from one symbol to the next.  If the output is modeled as a tape, the algorithm needs to ``remember'' where it is on the tape, but an instantaneous protocol makes no additional use of this information.

Our algorithm requires a memory register whose size grows as $\log N$.  However, \emph{any} protocol that emits uncorrelated asymptotically perfectly random bits \emph{and} achieves the Shannon bound must have a memory that grows with $N$.  Of the $NH(p)$ bits of entropy associated with the first $N$ input bits, $\sim\log N$ bits are associated with the Hamming weight, and cannot yet be distilled into perfectly random bits.  This entropy must be either:
\begin{enumerate}
\item written down on the output tape,
\item discarded, or
\item kept in memory until (with the addition of subsequent bits) it becomes distillable.
\end{enumerate}
The first solution ensures that some output bits are not perfectly random.  The second solution prohibits achieving the Shannon bound.  The third solution requires a memory whose size grows as $O(\log N)$ (and a non-Markovian protocol).

Our protocol is reversible, so it discards no entropy at all.  It therefore not only achieves the Shannon bound, but does so very tightly -- the total amount of randomness extracted from $N$ input bits is $NH(p) - O(\log N)$, which follows immediately from reversibility and the bounded size of the memory.  Furthermore, it also efficiently extracts purity, which in certain circumstances may be more useful than randomness.  We note that the Schulman-Vazirani cooling algorithm \cite{Schulman99} is also constructed from classical randomness extraction protocols, but is not streaming as it makes use of Peres' iterative von Neumann protocol.  Again, because the memory for our algorithm is so small, we know that on average $N(1-H(p))$ pure bits will be ejected.  Note that all of these figures are average values -- in any given experiment, the yield of random and pure bits will fluctuate by $O(\sqrt{N})$.

\subsection{Extracting entanglement}

This reversible protocol for extracting random bits can be adapted rather easily for entanglement concentration.  The only extra necessity is that Alice and Bob must implement the protocol not just reversibly, but also coherently (i.e., on a quantum information processor \footnote{We use ``quantum information processor'' rather than ``quantum computer'' because a quantum computer is implicitly scalable.  A nonscalable device limited to 100, 10, or even just 2 qubits is a quantum information processor \cite{RBKFoundPhys02}.  Our algorithm requires only $O(\log N)$ qubits to process $N$ input bits.  So, a quantum information processor comprising about 30 qubits (within the reach of current ion trapping technology) could concentrate a kilobit of entanglement -- and one with 100 qubits could concentrate over a terabit of entanglement.}).  The data registers must be quantum registers that can support superposition states without decohering, and the logic gates must preserve quantum superposition.  Moreover, each ``if-then'' statement in the algorithm must be implemented as a controlled operation, e.g. a quantum CNOT gate, rather than involving a measurement and conditioning on that measurement.

Suppose that a source produces pairs of systems one at a time in the joint (Alice-Bob) state
$$\ket\psi = \alpha\ket{0_A0_B}+\beta\ket{1_A1_B}.$$
The \emph{reduced} state of a single qubit on either Alice or Bob's side is
$$\rho = |\alpha|^2\proj{0} + |\beta|^2\proj{1}.$$
Suppose Alice and Bob each run our protocol coherently on their streams of qubits.  After $N$ input bits have been read, our streaming protocol has implemented Elias's block protocol on them.  Therefore, by Theorem \ref{thm:EliasConcentration}, it outputs perfect EPR pairs when performed coherently.  However, it is also instructive to consider why each output pair, considered individually, is maximally entangled.

Locally, Alice and Bob will each see output streams of maximally mixed qubits,
$$\rho_{\mathrm{out}} = \frac12\proj{0} + \frac12\proj{1}.$$
To show that all the entropy comes from entanglement -- i.e., Alice's $n$th output qubit forms an EPR pair with Bob's $n$th qubit -- let us consider just the first output qubit.
\begin{enumerate}
\item Alice's and Bob's input bits are perfectly correlated, and since the computational paths of their algorithms depend only on these input bits, their first output bit is perfectly correlated as well.  That is, if we were to measure Alice's first qubit and find it in the $\ket{0}$ state, then we would surely find the same result if we measured Bob's first qubit.
\item The algorithms that Alice and Bob run are completely reversible.  They involve no measurements and no outside randomness.  Furthermore, their joint input states are pure and thus carry no entropy at all.  Thus, given that their first output bits are perfectly correlated, these bits must form an EPR pair \emph{unless} they are decohered by some other system.  Such a system would have to be correlated with Alice or Bob's first output qubit, and it would have to be either another output qubit or a qubit still stored in memory.
\item The algorithm can be configured (as discussed above) to pause after outputting exactly one bit.  Thus, we can consider the first output qubit when there are no other output qubits, and so we can rule out the possibility that the first EPR pair is decohered by another output qubit.
\item The memory registers of both Alice and Bob's processors are uncorrelated with the state of the first output bit.  This follows quite simply from the way we built the protocol:  each path that outputs $\ket{0}$ is balanced with another path of the same length and the same probability that outputs $\ket{1}$.  Furthermore, by outputting a qubit, the protocol explicitly forgets which path it traversed.  Thus, while Alice and Bob's processors are each in a complicated superposition of different computational basis states (and are in fact highly entangled with each other), neither is even slightly correlated with the value of the first output bit.
\end{enumerate}
This shows that Alice and Bob's first output qubits form an EPR pair.  This EPR pair is utterly uncorrelated with anything else, particularly the memories of Alice and Bob's processors.  It follows that when Alice and Bob distill out their second qubits, they too are perfectly correlated with each other, and uncorrelated with anything else -- and therefore form an EPR pair, as do all subsequent pairs.

We conclude this section by pointing out a limitation of the algorithm presented so far.  It's basically a classical algorithm, adapted to run on a quantum computer in the computational basis.  Thus, it assumes and relies upon Alice and Bob's input states being diagonal in the computational basis.  Of course, if the input states were instead
$$\ket\psi = \alpha\ket{++}+\beta\ket{--},$$
then we could modify the algorithm very simply -- just perform an $SU(2)$ rotation on each input qubit to change the Schmidt basis.  However, we \emph{must} know the Schmidt basis of the input states.  Our algorithm (as presented so far) is a streaming implementation of the protocol originally introduced by Bennett et al in 1996 \cite{BennettPRA96}.  In Section \ref{SecSchur}, however, we show how to lift this requirement, constructing an algorithm for truly universal streaming entanglement concentration, which doesn't require any advance knowledge of the joint state (except a promise that it's pure).

\subsection{Computational complexity}

Let us now consider the resources necessary to implement our protocol.  One of the main advantages of a streaming protocol over a block protocol is reduced memory usage.  The streaming protocol doesn't need to store the entire input block of $N$ qubits!  Instead, our protocol requires three integer registers for $N$, $T$, and $L$.  Each register must be fully quantum (i.e., capable of storing arbitrary superpositions of integers), but only $\log(N)$ bits in size, since both $T$ and $L$ are less than or equal to $N$.

Our algorithm also requires some temporary storage to calculate its transitions between memory states.  Most of this calculation is trivial and can be done using $O(1)$ qubits.  The one major exception is calculating $\binom{N}{T}_L$.  Each iteration of the algorithm has to calculate the $L$th bit of two binomial coefficients.  This is nontrivial.  In fact, at first glance it looks almost impossible, since $T$ is typically $O(N)$, and
$$\binom{N}{T} = \frac{N!}{T!(N-T)!}$$
is $O(N)$ bits in size.  Calculating it involves $O(N)$ multiplications and divisions of integers with $O(\log N)$ bits each, and storing the result requires $O(N)$ bits of memory.

Fortunately, we only need to compute a single bit of $\binom{N}{T}$.  This removes any need to store a number with $O(N)$ bits.  We can then take either of two routes (depending on which is more convenient) to run the algorithm in $O(\log N)$ qubits of memory.
\begin{enumerate}
\item We can run the entire algorithm -- including computing bits of binomial coefficients -- on a quantum processor, with no classical assistance at all.  This turns out to be possible because computing the $L$th bit of $\binom{N}{T}$ is in the complexity class LOGSPACE.  Thus, temporary memory requirements can be held to $O(\log N)$.  However, this makes the algorithm design much more complicated, and may slow it down substantially (since we trade time for space).
\item We can precompute the binomial coefficients with a classical processor.  If we have $\mathrm{poly}(N)$ classical memory, then this can be done relatively quickly, and the results used to implement the quantum protocol.  The trick here is that the classical computer cannot know the values of $N$, $T$, and $L$ -- if it did, it would decohere the computation.  So the classical computer has to calculate \emph{all} of the $O(N^2)$ possible binomial coefficients.  Though clumsy, this approach is probably more practical for moderate $N$, and minimizes the amount of quantum computation necessary.
\end{enumerate}

Computing binomial coefficients is in LOGSPACE because division and iterated multiplication are both in LOGSPACE \cite{HesseJCSS02}.  The quotient of two $N$-bit numbers, or the product of $N$ $N$-bit numbers, can be computed in $O(\log N)$ space.  $N!$ is the product of $N$ numbers whose size is $\log N$ bits, so we can compute it in LOGSPACE.  Three such computations yield $N!$, $T!$, and $(N-T)!$, and computing the binomial coefficient involves two divisions.

This may seem paradoxical -- how can an $N$-bit number be computed in $O(\log N)$ bits of space?  We are allowed a machine with $O(\log N)$ read-write memory, plus an unbounded read-only tape containing the problem specification (e.g., the $N$-bit numbers to be divided, or the $N$ numbers to be multiplied), and an unbounded write-only tape on which the answer will be written out.  This model is very adaptable to our problem.  To calculate the $L$th bit of $\binom{N}{T}$, we chain three such machines together.

\begin{center}\includegraphics[width=2.5in]{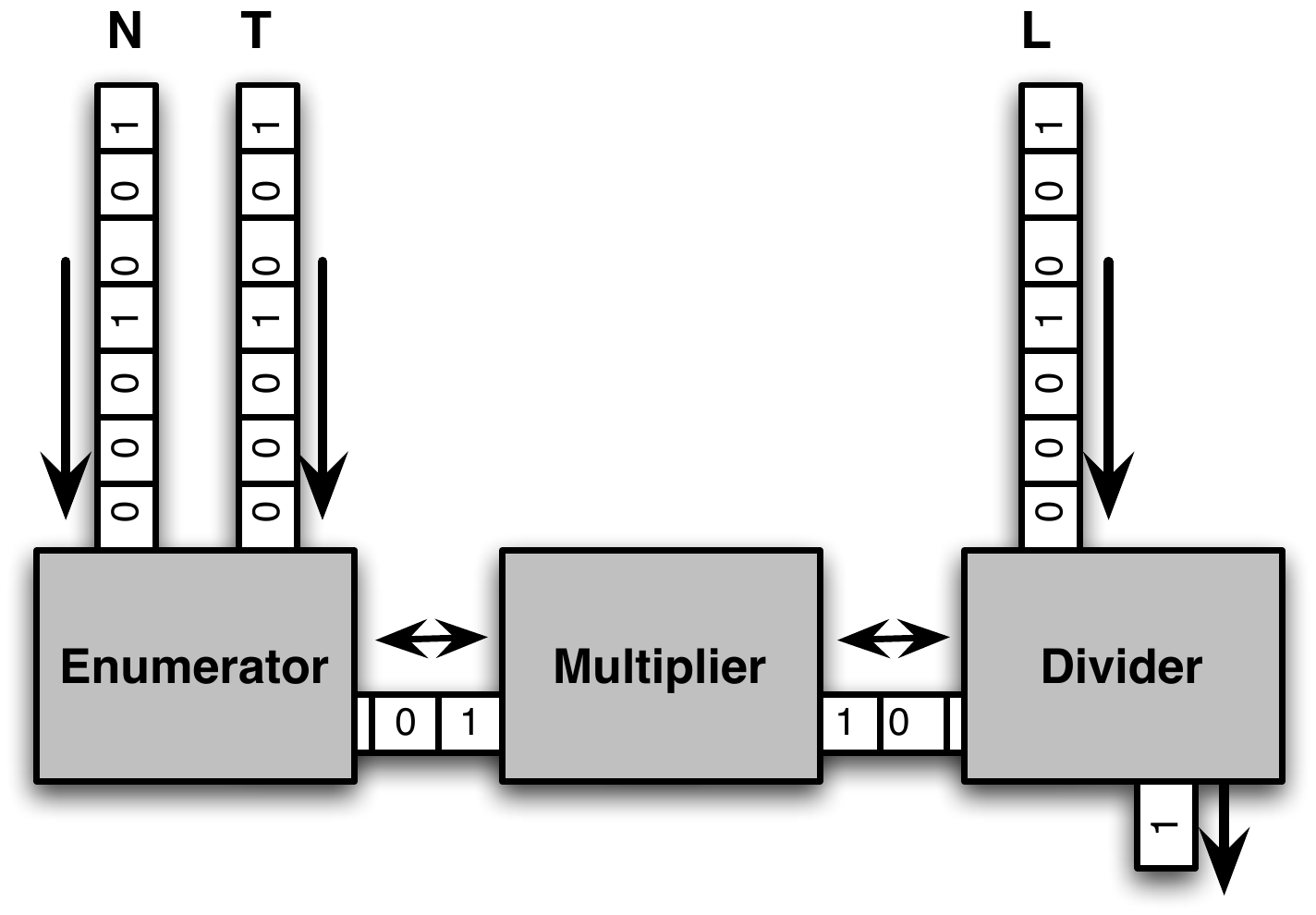}\end{center}

The first has a $\log N$-sized input tape containing $N$ and $T$.  It constructs (and passes to the second machine) two long lists of integers: $\{1\ldots N\}$ and $\{1\ldots T,1\ldots N-T\}$.  This is easy to do with $O(\log N)$ memory.  The second machine multiplies together the numbers in these lists, and passes the results to the third machine as two integers, $N!$ and $T!(N-T)!$.  The third machine divides these numbers, calculating only the $L$th bit of the result, and outputs $\binom{N}{T}_L$.

Communication between the machines is accomplished via queries.  Instead of reading a long read-only tape, the second and third machines tell their predecessor which bit of a ``virtual tape'' they need, and the predecessor computes it on the fly.  This trades time for space, avoiding the need for an $O(N\log N)$ memory tape, at the cost of extra time complexity.

We do not know the time complexity of this approach, but it seems unlikely to be low.  Ideally, a streaming protocol would process each input symbol in $O(1)$ time, processing all $N$ symbols in $O(N)$ time.  This is manifestly impossible for an adaptive protocol, which has to maintain and process some record of what it's read so far.  The size of that record grows as $O(\log N)$,  which suggests a lower bound of $\Omega(N\mathrm{polylog}N)$ for processing $N$ symbols (since processing the $N$th symbol involves a polynomial-sized computation on a memory of size $\log N$).

We do not know whether this can be achieved, but the straightforward approach given above certainly doesn't.  Processing the $N$th symbol involves computing $\binom{N}{T}_L$, and doing this in with minimal space takes at least $O(N\mathrm{polylog} N)$ time.  This is itself only a lower bound.  Because the various stages in the computation query each other, we can only say that the time complexity of computing $\binom{N}{T}_L$ this way is $O(\mathrm{poly}(N))$.  

Reducing the per-symbol cost from $O(\mathrm{poly}(N))$ to $O(\mathrm{polylog} N)$ would make our protocol much more useful in practice.  To do so, we need to avoid computing \emph{all} the bits of $N!$ to get just one bit of $\binom{N}{T}$.  However, a similar problem -- computing $N!\ \mathrm{mod}\ P$ (where $P>N$) -- is thought to be hard.  An $O(\mathrm{polylog}(P))$ algorithm would yield an efficient algorithm for integer factoring.  So finding fast ways to exactly compute bits of binomial coefficients, while theoretically interesting, is probably not the best way to go about this.  A more promising approach is to approximate $\binom{N}{T}$ to fixed (or $O(\log N)$) precision, e.g. with Stirling's approximation.  Because $L$ is exponentially distributed, $L$ will almost always be very close to $L_{\mathrm{max}} = \lfloor \log_2\binom{N}{T}\rfloor$, and so computing only the most-significant $K$ bits of $\binom{N}{T}$ should induce an error of at most $2^{-K}$.  Formally, this precludes actually achieving the Shannon bound -- but in practice, such tiny (and controllable) deviations are insignificant.  A similar approach is almost always used in arithmetic coding, where the use of finite precision reduces computational complexity at the price of a tiny loss in compression efficiency.

A more practical approach is to offload as much computation as possible onto a classical computer.  While not necessarily a good long-term strategy, this is a promising solution as long as quantum memory is limited and precious.  To process the $N$th bit this way, we use the classical computer to loop over \emph{every} value of $T$ and $L$.  It computes $\binom{N}{T}_L$, uses this to design a unitary circuit, and then applies that unitary \emph{conditional} on $\proj{T,L}$.  Since there are at most $N$ possible values of $T$ and $L$, and the unitary acts on a register of size $O(\log N)$, the time required to process a single input symbol is $O(N^2\mathrm{polylog(N)})$.  This is undeniably ugly, but provides a simple constructive approach to implementing our algorithm in a bounded amount of time.

\section{A fully quantum protocol:  the streaming Schur transform}\label{SecSchur}

The algorithm that we presented in the previous section requires Alice and Bob to know something about the state $\ket{\psi}$ describing their systems.  Specifically, they need to know the Schmidt basis, which we've written (without loss of generality) as $\{\ket{0},\ket{1}\}$, where
\begin{equation}
\ket{\psi} = \alpha\ket{0_A0_B}+\beta\ket{1_A1_B}.
\end{equation}
It is this knowledge of the Schmidt basis that reduces the problem to classical randomness extraction.  Note, however, that Alice and Bob do \emph{not} need to know $\alpha$ and $\beta$.  Our protocol is classically universal (i.e., independent of the probabilities $|\alpha|^2,|\beta|^2$), but not quantumly universal.  In this section, we fix this problem and generate a completely universal streaming protocol, by incorporating the quantum Schur transform.  The resulting algorithm is a streaming implementation of Matsumoto and Hayashi's optimal block concentration protocol \cite{MatsumotoPRA07}.

\subsection{Quantum types, representation theory, and Schur-Weyl duality}

The algorithm that we developed in previous sections performs a particular transformation on strings.  It divides the $N$-bit input string into a permutation-invariant \emph{type}, and an index $\alpha\in\left[0\ldots\binom{N}{T}-1\right]$ into the type class.  (All the complicated business with $L$ is necessary only because we want to efficiently convert $\alpha$ into random bits).  This transformation is used frequently in classical information theory, where it gives rise to the \emph{method of types}.  Its usefulness arises because we are dealing with a permutation-invariant distribution over input strings, so separating out the permutation-invariant part is handy.  Furthermore, the index $\alpha$ isn't just permutation-dependent; it's \emph{uniformly} random when the input is permutation-invariant.  This is because the permutation group $S_N$ acts transitively on type classes -- i.e., given any two strings $s,s'$ in a type class, there is a permutation that transforms $s\to s'$.

In the absence of a preferred basis, we can't apply the classical method of types directly.  Instead, our algorithm must deal with arbitrary vectors in the Hilbert space of $N$-qubit quantum strings, $\mathcal{H} = \left(\mathcal{H}_2\right)^{\otimes N}$.  Fortunately, there is an analogous method of quantum types \cite{HarrowThesis}, and a corresponding transformation on quantum strings that divides them into a permutation-invariant ``type'' and an ``index'' into that type class.  This transformation is the Schur transform \cite{BaconPRL06}, and after introducing it in this section, we'll show how to combine it with our randomness-distillation protocol.

We can apply permutations to $N$ qubits, just like $N$ classical bits.  Each of the $N!$ permutations in the symmetric group $S_N$ is represented by a $2^N\times 2^N$ unitary operator acting on $\mathcal{H}$.   These operators form a \emph{representation} of $S_N$.  This representation is reducible, meaning that $\mathcal{H}$ can be divided into a direct sum of subspaces $\mathcal{H}_k$, each closed under the action of every permutation in $S_N$.  These subspaces are \emph{irreducible} representation spaces, a.k.a ``irreps'', of $S_N$, and they are the quantum equivalent of type classes.

The analogy between classical and quantum types is not as straightforward as one might think from the previous paragraph.  To see this, let's consider the simplest possible example:  two qubits.  Their Hilbert space is $\complex^4$, and the permutation group $S_2 = \{\Id, \pi_{(12)}\}$ has two elements.  Since $\Id$ acts trivially on all states, the irreps of $S_2$ are the eigenspaces of $\pi_{(12)}$.  Its eigenvalues are $\{+1,-1\}$, and its action on $\complex^4$ defines two invariant subspaces:  a 1-dimensional \emph{antisymmetric} subspace (the ``singlet''),
\begin{equation}
\mathcal{H}_{\mathrm{antisymmetric}} = \mathrm{Span}\left(\frac{\ket{01}-\ket{10}}{\sqrt2}\right)
\end{equation}
and a 3-dimensional \emph{symmetric} subspace (the ``triplet'')
\begin{equation}
\mathcal{H}_{\mathrm{symmetric}} = \mathrm{Span}\left(\left\{\ket{00},\frac{\ket{01}+\ket{10}}{\sqrt2},\ket{11}\right\}\right).
\end{equation}
The singlet is an irreducible representation space.  The triplet, however, is not irreducible -- in fact, \emph{any} proper subspace of the triplet is itself invariant, since both elements of $S_2$ act trivially on it.  If we try to reduce the triplet to a direct sum of irreps, we face an embarrassment of choices -- there is no preferred decomposition into 1-dimensional subspaces.

Just for contrast, consider the classical case of two bits.  There are three type classes:  $\{00\}$, $\{01,10\}$ and $\{11\}$.  Each is invariant under permutations, and ``irreducible'' (meaning that it cannot be further subdivided).  Both $00$ and $11$ are symmetric strings, but they are distinguished from one another by their Hamming weight, and by the existence (in classical theory) of a preferred set of symbols, $\{0,1\}$.  If we chose $\{\ket{0},\ket{1}\}$ as a preferred basis for qubits, we could use it to divide the triplet into irreducible subspaces spanned by $\left\{\ket{00},\frac{\ket{01}+\ket{10}}{\sqrt2},\ket{11}\right\}$.  However, the breaking of unitary symmetry is arbitrary and unsatisfying.

That very unitary symmetry suggests a much more elegant solution.  The triplet and singlet are each invariant, not only under permutations, but also under \emph{collective} unitary rotations.  That is, we apply the same $U\in SU(2)$ to each qubit.  Collective rotations of the form $U\otimes U$ (or $U^{\otimes N}$ in general) are a representation of the group $SU(2)$, and the singlet and triplet (being invariant under these rotations) are representation spaces.  Furthermore, they are both \emph{irreducible} representation spaces, for they have no proper rotation-invariant subspaces.

This is the simplest example of \emph{Schur-Weyl duality}.  Schur-Weyl duality is the statement that, given a Hilbert space $\mathcal{H}_d^{\otimes N}$:
\begin{enumerate}
\item The action of the symmetric group $S_N$ commutes with the action of the collective rotation group $SU(d)$, and
\item $\mathcal{H}_d^{\otimes N}$ decomposes into a direct sum of subspaces $\mathcal{H}_\lambda$, each of which is the direct product of an irrep of $SU(d)$ with an irrep of $S_N$:
\begin{equation}\label{eqSchurWeylDecomposition}
\mathcal{H}_d^{\otimes N} = \bigoplus_\lambda{ \mathcal{U}_\lambda\otimes\mathcal{P}_\lambda }
\end{equation}
\end{enumerate}
For two qubits, there are two terms in the decomposition, which we'll denote $\lambda=0,1$, so:
\begin{equation}
\mathcal{H}_2^{\otimes 2} = \mathcal{U}_0\otimes\mathcal{P}_0 \oplus \mathcal{U}_1\otimes\mathcal{P}_1
\end{equation}
Both representations of $S_2$ are trivial, so $\mathcal{P}_0$ and $\mathcal{P}_1$ are both 1-dimensional.  The triplet ($\mathcal{U}_1$) is a 3-dimensional irrep of $SU(2)$, while the singlet ($\mathcal{U}_0$) is 1-dimensional.  We need to add a third qubit to obtain a nontrivial symmetric group representation:  the action of $S_3$ on three qubits has two irreps, one of which is 2-dimensional.

In this decomposition of $N$-qubit strings, the $\mathcal{P}_\lambda$ spaces correspond to type classes, while the irrep label $\lambda$ and the $\mathcal{U}_\lambda$ spaces \emph{together} correspond to the classical type.  This is a little confusing at first; why do we need \emph{two} variables to describe the ``type'' of a quantum string?  It makes more sense if we look at classical types in a slightly different way.  First, we note that whereas the reversible transformations on a single qudit are unitaries in $SU(d)$, the corresponding transformations on a classical $d$-ary system are elements of $S_d$ -- i.e., permutations of the $d$ symbols.  So, we can divide a classical type $T = \{n_1\ldots n_d\}$ into two parts:  (1) a \emph{sorted} list of frequencies $\tilde T=\{n_1\geq n_2 \geq \ldots n_d\}$, and (2) a permutation in $S_d$ identifying which of the $d$ symbols appears 1st, 2nd, etc. in the sorted list.  This view of classical types turns out to be \emph{exactly} analogous to the Schur-Weyl decomposition.  The irrep labels $\lambda$ correspond precisely to nonincreasing partitions $\{n_1\geq n_2\geq\ldots n_d\}$ (where $\sum_k{n_k} = N$).  These ``frequencies'' relate to the eigenvalues of $\rho$ in exactly the same way that the type of an $N$-bit string of i.i.d.~symbols relates to the source probabilities -- if $\rho$ has eigenvalues $\{p_k\}$, then as $N$ gets large, measuring the irrep label of $\rho^{\otimes N}$ gives $\{n_k\} \approx \{Np_k\}$ with high probability.  The $\mathcal{U}_\lambda$ spaces carry information about the diagonal basis of $\rho$.

\subsection{Applying quantum types to concentration}

If Alice and Bob share $N$ partially-entangled qubit pairs in state $\ket{\psi}$, they describe their respective systems by $\rho_A^{\otimes N}$ and $\rho_B^{\otimes N}$, where $\rho_A$ and $\rho_B$ are partial traces of $\proj{\psi}$.  Because $\rho^{\otimes N}$ is permutation-invariant, it can be decomposed according to Eq. \ref{eqSchurWeylDecomposition} as
\begin{equation}
\rho^{\otimes N} = \sum_\lambda{ p_\lambda \rho_\lambda \otimes \frac{\Id}{\mathrm{dim}(\mathcal{P}_\lambda)}},
\end{equation}
a state that is maximally mixed over each type class $\mathcal{P}_\lambda$.  This follows from Schur's Lemma; if a matrix $\rho$ is invariant under $\rho\to\pi\rho\pi^\dagger$ for all $\pi$ in a representation $G$, then $\rho$ is a direct sum of scalar matrices on the irreps of $G$.  The conditional states $p_\lambda\rho_\lambda$ on the various $SU(2)$ irreps are determined by $\rho$, and aren't especially relevant to this discussion.

This is the quantum counterpart of the classical observation that i.i.d.~distributions of strings are uniformly distributed within type classes.  
For the purposes of entanglement concentration, the states on Alice's $\mathcal{P}_\lambda$ subspaces are not just uniformly random.  They are maximally entangled with their counterparts on Bob's side.  So if Alice and Bob each measure $\lambda$, they get identical results $\hat\lambda$, and are left with a state
\begin{equation}
\rho_A =\rho_B = \rho_{\hat\lambda}\otimes \frac{\Id}{\mathrm{dim(\mathcal{P}_{\hat\lambda})}}.
\end{equation}
Now, recall that they started with a pure state $\ket{\psi}^{\otimes N}$, and performed a projective measurement.  This means that their post-measurement state is pure -- and thus their maximally mixed reduced states correspond to a maximally entangled pure state
\begin{equation}
\hspace{0.4in}\ket{\psi} = \ket{u_{\hat\lambda}}
\otimes \left(\frac{1}{\sqrt{\mathrm{dim}(\mathcal{P}_{\hat\lambda})}}\sum_{j=1}^{\mathrm{dim}(\mathcal{P}_{\hat\lambda})}{\ket{j_Aj_B}}\right).
\end{equation}
If they want perfect EPR pairs, they may as well discard the $\mathcal{U}_{\hat\lambda}$ subsystem, which contains $O(\log N)$ bits of non-maximal entanglement.  They are left with a maximally entangled state over the entire $\mathcal{P}_{\hat\lambda}$ subspace.  This can be converted into EPR pairs by partitioning $\mathcal{P}_{\hat\lambda}$ into subspaces of dimension $2^L$ and measuring $L$, exactly as explained in the proof of Theorem \ref{thm:EliasConcentration}.

\subsection{The streaming Schur transform}

Matsumoto and Hayashi showed how to use the Schur-Weyl decomposition (Eq. \ref{eqSchurWeylDecomposition}) and its properties to achieve optimal universal compression \cite{HayashiPRA02} and entanglement concentration \cite{MatsumotoPRA07}.  These are non-constructive information-theoretic results, like Shannon's random-coding proof of channel capacity, rather than practical implementations.  However, Bacon et. al. recently demonstrated a quantum algorithm to perform the \emph{quantum Schur transform}, which points the way to implementing these protocols efficiently on a quantum computer \cite{BaconPRL06}.  Our goal in this section is to use the Bacon et. al. algorithm as a building block for a \emph{streaming} concentration/compression protocol.

The Schur transform transforms an $N$-qubit Hilbert space $\mathcal{H}_2^{\otimes N}$ into the direct-sum Hilbert space given in Eq. \ref{eqSchurWeylDecomposition},
\begin{equation*}
\mathcal{H}_d^{\otimes N} = \bigoplus_\lambda{ \mathcal{U}_\lambda\otimes\mathcal{P}_\lambda }.
\end{equation*}
This is just a change of basis -- but, then, every unitary transformation is ``just'' a change of basis.  The Schur transform takes as input a single $N$-qubit register, and outputs three quantum registers of different sizes.  We'll call these registers $T$, $U$, and $P$, and in the following list we describe each register and give an example of what its state would be for an input string $\rho^{\otimes N}$.
\begin{enumerate}
\item The $T$ register holds the irrep label $\lambda$.  It is spanned by a basis $\{\ket{\lambda}: \lambda=0\ldots\left\lceil\frac{N+1}{2}\right\rceil\}$.  Measuring the $T$ register provides the best possible estimate of $\rho$'s spectrum -- i.e., whether the individual qubits of the input state are consistently aligned along a particular direction in $\mathcal{H}_2$.
\item The $U$ register holds the state of the $SU(2)$ irrep $\mathcal{U}_\lambda$.  The dimension of $\mathcal{U}_\lambda$ depends on $\lambda$, so $U$ has to be big enough to hold the largest $\mathcal{U}_\lambda$, which is $(N+1)$-dimensional.  $U$ is spanned by a basis $\{\ket{m}: m=0\ldots N\}$.  Measuring the $U$ register provides the best possible estimate of the eigenbasis of $\rho$ -- which, for qubits, is equivalent to the direction of its Bloch vector.  Unlike the $T$ register, the $U$ register does not have an unique basis in which we would measure it to extract information.  Measuring the $\{\ket{m}\}$ basis yields the best estimate of the input string's Hamming weight in the $\{\ket{0},\ket{1}\}$ basis, but if we wanted to know its Hamming weight in the $\{\ket{+},\ket{-}\}$ basis, a different measurement would be optimal.
\item The $P$ register holds the state of the $S_N$ irrep.  As with $U$, this register must be large enough that we can embed \emph{any} of the $\mathcal{P}_\lambda$ spaces into it.  In fact, it must be at least $\frac{2^N}{O(N^2)}$-dimensional, because we're mapping $\mathcal{H}_2^{\otimes N}$ into $T\otimes U\otimes P$, yet both $T$ and $U$ are $O(N)$-dimensional.  In the Bacon et. al. implementation, the $P$ register comprises exactly $N$ qubits, denoted $\{p_1,p_2,\ldots p_N\}$.  When we Schur-transform $\rho = \rho^{\otimes N}$, measurements on this register yield random results.
\end{enumerate}
The key ingredient in the Schur transform is the \emph{Clebsch-Gordan transform}.  It takes as its input the $T$ and $U$ registers, along with the $n$th qubit $s_n$, and outputs updated $T$ and $U$ registers along with the $n$th qubit of the $P$ register, $p_n$.

\begin{center}\includegraphics[width=2in]{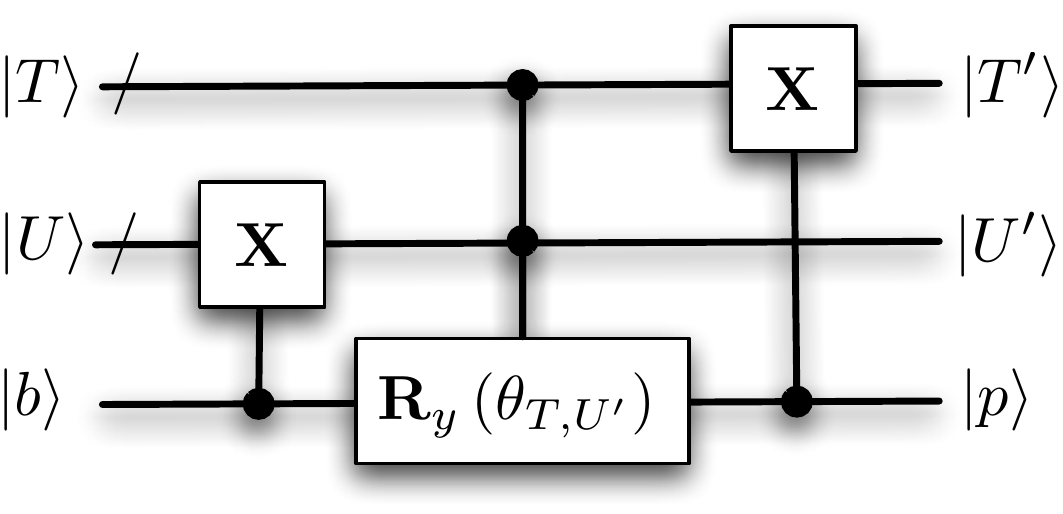}\end{center}

The full Schur transform then consists of initializing the $T$ and $U$ registers, then sequentially applying Clebsch-Gordan transforms to each of the $N$ input qubits:

\begin{center}\includegraphics[width=\HPW]{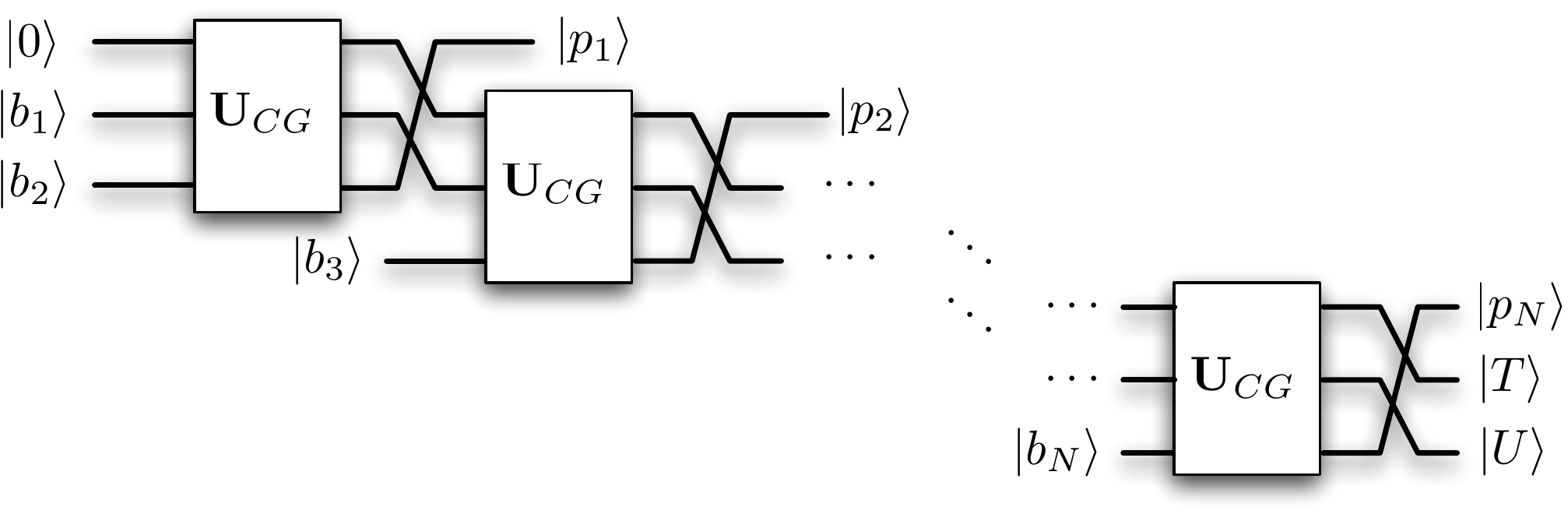}\end{center}

Just a brief glance at the circuit above shows that this implementation of the Schur transform is appropriate for a streaming protocol.  It addresses the input qubits one at a time, and never reuses an earlier qubit.  The only problem is that the $P$ register, holding the $S_N$ irrep, is not in the right form.  Actually, this is a fairly serious problem for \emph{any} application to concentration or compression, because the $P$ register comprises $N$ qubits -- no matter what the input is.  For each input qubit $s_n$, exactly one $p_n$ gets emitted, so the entropy of the input qubits is uniformly distributed across the $N$ $\{p_n\}$ qubits, rather than being compressed.

Compressing the $P$ register requires a peek at the representation theory of $S_N$.  As we mentioned above, the irreps of $S_N$ are labeled by an index $\lambda$, whose values are in 1:1 correspondence with nonincreasing sequences of at most $d$ integers, $\{n_1\geq n_2 \geq \ldots \geq n_d\}$ where $\sum_k{n_k}=N$.  These sequences are usually depicted by \emph{Young diagrams}, arrays of $N$ boxes in at most $d$ rows, with $n_k$ boxes in their $k$th row.  Here is the Young diagram for an irrep of $S_6$:

\begin{center}\includegraphics{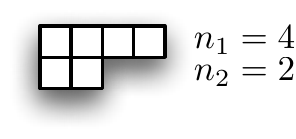}\end{center}

Because the diagram has only 2 rows, it labels irreps of $S_6$ and $SU(2)$ which appear in the decomposition of $\mathcal{H}_2^{\otimes 6}$.  Diagrams with more than 2 rows are not relevant to qubits; they label valid representations of $S_N$, but not of $SU(2)$.  This particular diagram corresponds (roughly) to the class of strings with 4 qubits aligned along a common axis and 2 qubits aligned \emph{against} that axis.

Now, when the Schur transform circuit addresses the $N$th input qubit, $N-1$ qubits have already been transformed.  The state of the $T$ register is therefore a superposition or mixture of states corresponding to Young diagrams with $N-1$ boxes (i.e., $\ket{\lambda=\{n_1,N-1-n_1\}}$.  Adding another qubit corresponds to adding another box to the Young diagram.  We can add it to the first row, or add it to the second row \emph{if} the second row isn't already as long as the first row.  As we read more qubits in, the $T$ register (following this rule) traverses \emph{Young's lattice}:

\begin{center}\includegraphics[height=2.2in]{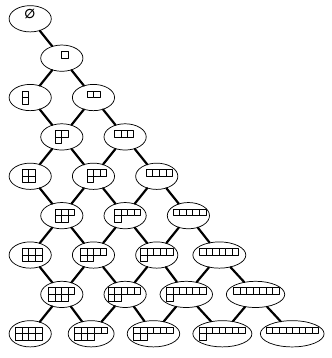}\end{center}

This looks quite a bit like Pascal's triangle, and it plays exactly the same role.  Different types correspond to different locations in the lattice (plus, in the quantum case, an $SU(2)$ register that's not shown), while different strings within a type class correspond to distinct \emph{paths}.  In Young's lattice, each node is labeled by a Young diagram, which labels an irrep of $S_N$ (i.e., a quantum type class).  Each of the paths to a given node corresponds to a distinct state within that class.  Thus, by counting the paths to a node, we obtain the dimension of each representation space:

\begin{center}\includegraphics[height=2.2in]{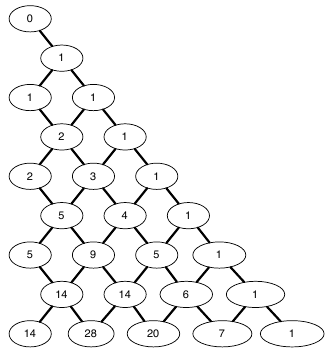}\end{center}

The representation spaces don't have a unique basis, but the path-counting procedure above suggests a convenient basis known as Young's orthogonal basis.  We simply assign to each path $p$ a basis state $\ket{p}$.  Paths to a node in the $N$th row of Young's lattice consist of $N$ steps, and each step is either to the right (meaning we add a box to the first row of the Young diagram) or to the left (meaning we add a box to the second row).  Clearly, any path can be denoted by a sequence of $N$ symbols from the set $\{L=\mathrm{left},R=\mathrm{right}\}$, e.g. $p = RRLLR\ldots$, and thus we can encode all such paths into $N$ bits -- or, since we are dealing with quantum strings, and can traverse Young's lattice in superposition, into $N$ qubits.

This encoding is not efficiently compressed, nor is it appropriate for entanglement concentration.  Since the $P$ register contains complete information about the path taken through Young's lattice, it also contains information about the end-point of the path -- i.e., about the irrep label stored in $T$.  Moreover, for strings in high-weight irreps, most of the steps will be to the right, so most of the $p_k$ bits will be ``R''.  We need to compress the $P$ register in order to extract EPR pairs from it.

Our algorithm is almost perfectly suited to this.  In fact, it can be applied directly with only two changes:
\begin{enumerate}
\item Our algorithm traversed the lattice of Pascal's triangle, whose nodes' sizes are binomial coefficients.  We need to adapt it to traverse Young's lattice, whose nodes have different sizes.  The dimension of an irrep $Y$ of $S_N$ is given by the \emph{hook length formula}:
\begin{enumerate}
\item Draw the Young diagram.
\item To each of the $N$ boxes $x$ in the Young diagram, assign a ``hook length'' $h(x)$, which is the sum of (a) the number of boxes to the right of $x$; (b) the number of boxes directly below $x$; and (c) 1 for $x$ itself.
\item The size of $Y$ is given by
\begin{equation}
\mathrm{Size}(Y) = \frac{N!}{\prod_x{h(x)!~}~}.
\end{equation}
\end{enumerate}
A short calculation for the Young diagram with $N-T$ boxes in the first row and $T$ in the second row gives
\begin{equation}
\mathrm{Size(Y)} = \binom{N}{T} \frac{N-2T+1}{N-T+1}.
\end{equation}
So the size of a quantum type class is very nearly equal to the size of the corresponding classical type class, with a simple rational function giving the discrepancy.  Every instance of ``calculate a binomial coefficient'' in our original algorithm gets replaced by ``calculate the corresponding irrep dimension''.
\item Instead of performing operations conditional on the classical Hamming weight $T$, we condition our operations on the irrep label $T$.  Since all the operations in our algorithm are necessarily coherent anyway, this change brings no significant changes.
\end{enumerate}
This defines what we will call the \emph{quantum streaming Elias transform}.  A single step of the transform can be represented as a unitary operation $U_E$.  

\begin{center}\includegraphics[width=1in]{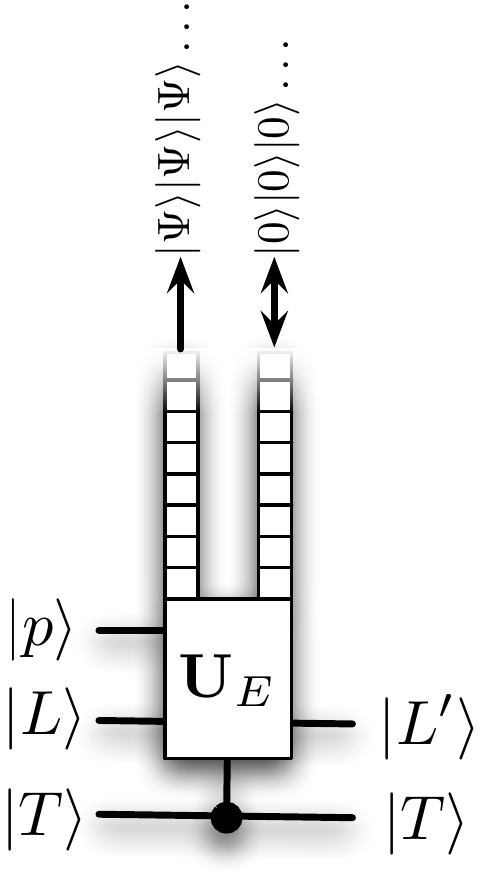}\end{center}

$U_E$ acts on two registers -- an $SU(2)$-invariant qubit $\ket{p}$ and the bin size $\ket{L}$ -- conditional on a third, the irrep label $\ket{T}$.  It also has access to two variable-length tapes.  One is output-only, and holds EPR pair halves.  The other is bidirectional, and holds pure $\ket{0}$ qubits.  The $\ket{p}$ qubit always goes out onto one tape or the other -- but sometimes, $U_E$ also pops one or more qubits off the purity tape, fills them with entanglement from the $\ket{T}$ and $\ket{L}$ registers, and pushes them out the EPR tape.

We can use this protocol, together with the Schur transform, to make a completely universal extraction protocol.

\begin{protocol}
\begin{center}\includegraphics[width=\HPW]{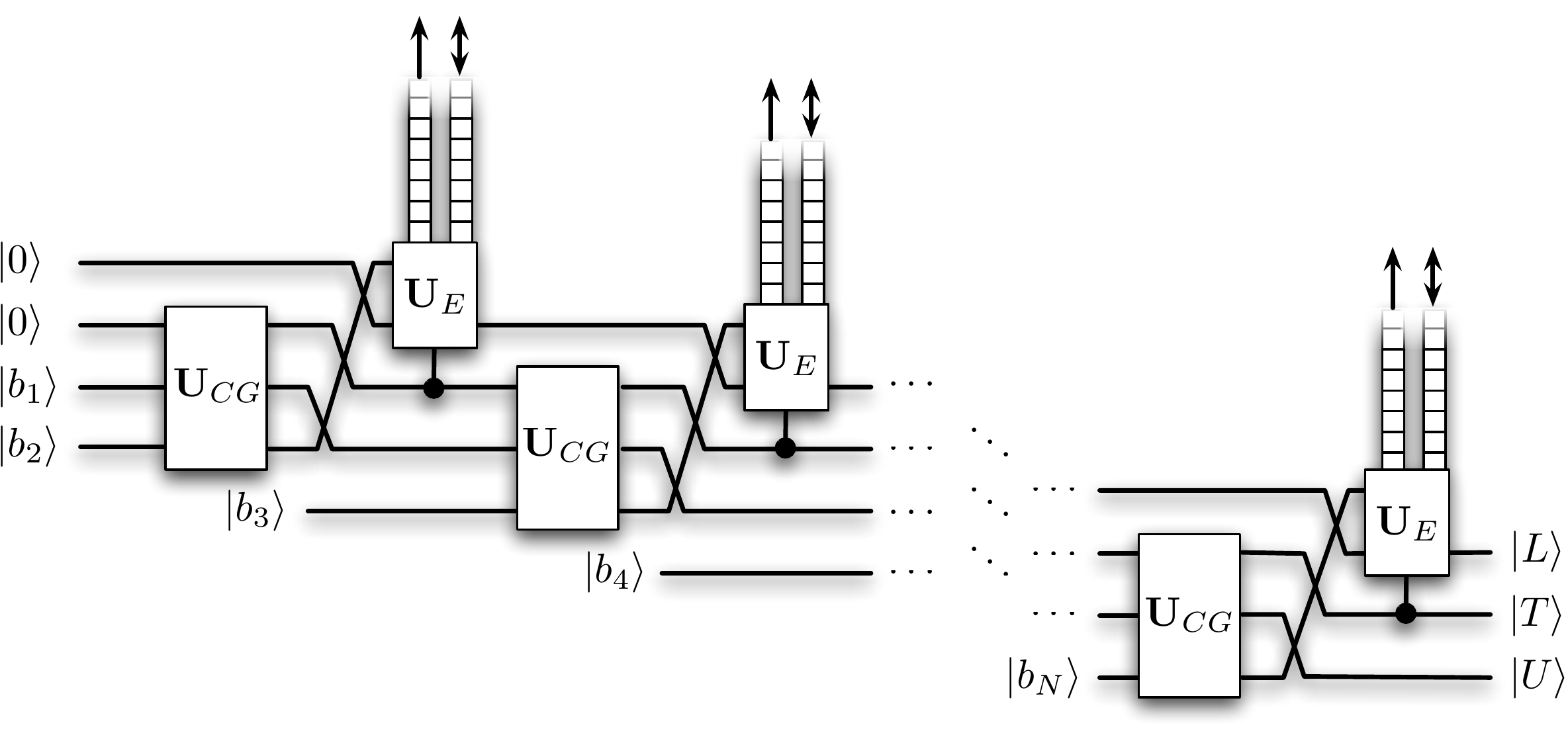}\end{center}
\begin{enumerate}
\item Each new qubit $\ket{b_n}$ is Clebsch-Gordan transformed, yielding updated $\ket{T}$ and $\ket{U}$ registers, and an $SU(2)$-invariant qubit $\ket{p_n}$.
\item $\ket{T}$ and $\ket{p_n}$ are fed into a quantum Elias transform, along with $\ket{L}$.
\item The physical input qubit $\ket{b_n}$, suitably transformed into either $0$ or $1$ EPR pair-half, emerges immediately on one of the two tapes.
\end{enumerate}
\end{protocol}

A few remarks are in order here.  Our algorithm is basically an interleaving of the Schur transform with the quantum Elias transform.  These two components are coupled only by $\ket{T}$; the $\ket{U}$ and $\ket{L}$ registers are only used by the Schur and Elias components (respectively).  We have described the protocol's fully streaming mode, where $N$ can be assumed classical.  On-demand mode requires another quantum register for $\ket{N}$.  The integer registers ($\ket{T},\ket{L},\ket{U}$) must grow with $N$.  If quantum memory is at a premium, pure qubits may be scavenged from the end of the purity tape.  However, the purity tape is \emph{also} used by $U_E$ as a source of fresh qubits, whenever it reads a single $\ket{p}$ bit and outputs more than one EPR pair-half.

\section{Discussion}

This is the first adaptive (streaming and universal) protocol for entanglement concentration.  Because it runs in very little space, it can be implemented using current (or near-future) technology.  This opens the door for experimental implementations of a variety of information theoretic protocols.  We have already used the ideas in this paper to design sequential protocols for \emph{optimal} quantum data compression and state discrimination, which use $O(\log N)$ or even $O(1)$ memory.

Although this protocol can be used for quantum data compression (details will be given elsewhere), good data compression algorithms can fail at entanglement concentration.  There are many ways to encode compressed data which do not meet the (more stringent) structure requirements for concentrated EPR pairs.  Reversible entanglement concentration, on the other hand, seems to necessarily yield data compression \footnote{However, reversibility is critical for compression. For instance, Bennett et. al.'s block concentration protocol \cite{BennettPRA96} fails at compression.  Alice and Bob make precise measurements of their strings' type, which does terrible damage to the input state.}.  Entanglement concentration seems to have stricter requirements than compression.  Given the role of compression in information theory, this suggests that more insights can be gained by applying the stricter requirements of concentration.

Our protocol bolts together two components.  It seems possible to regard \emph{either} component of the algorithm as ``trivial''.  From one perspective, the Schur transform does all the heavy quantum lifting; our algorithm just compresses the $P$ register.  However, consider the classical version of this protocol.  A classical Schur transform does the following:
\begin{enumerate}
\item It counts the number of ``1'' bits in the input to obtain the Hamming weight $T$.
\item It separates the type into two registers: (a) a single-bit ``dictionary'' $U$ that identifies whether $T$ or $N-T$ is bigger (i.e., whether 0 or 1 appears more often); and (b) a ``sorted type'' $\mathrm{max}(T,N-T)$.
\item It strips out the dictionary information ($U$), by replacing the $k$th input bit $s_k$ with a dictionary-invariant bit $p_k = s_k \oplus U$.  This ensures that the $\{p_k\}$ are invariant under any ``collective rotation'' of the entire string.\end{enumerate}
Most of this is computationally trivial.  The most significant step is adding up the Hamming weight of the input.  So the classical equivalent of the Schur transform is basically sequential addition -- which we took for granted in our implementation of the streaming Elias protocol!  Stripping the dictionary register $U$ out of the $\{s_k\}$, which seems optional (and, in fact, rather arbitrary) in the classical variant, is a necessary part of quantum sequential addition; the no-cloning theorem prohibits us from copying information, so in order to calculate and store it in $U$, we must remove all traces of it from the other registers.

The previous paragraph should \emph{not} be taken to imply that the Schur transform itself is in any way trivial.  Rather, we are suggesting that the Schur transform can be seen as the fully quantum analogue of sequential addition.  This isn't actually all that surprising, since the main application of Clebsch-Gordan coefficients is in the addition of angular momentum.  Nonetheless, there is a subtle distinction worth noting:  whereas Clebsch-Gordan coefficients are used to do \emph{classical} calculations about quantum systems, the Schur transform is a fully quantum physical operation.  A similar distinction divides classical simulation of a quantum system from quantum simulation of a quantum system.

Finally, our construction has implications for quantum learning.  Adaptive classical protocols are closely tied to machine learning. Our protocol demonstrates how a quantum computer can ``learn'' a quantum source, and adapt its strategy, without ever making a measurement or collapsing the input state.

\begin{acknowledgments}
The authors are supported by the Government of Canada through Industry Canada and by the Province of Ontario through the Ministry of Research \& Innovation.  D.G. was supported by CIFAR and NSERC.
\end{acknowledgments}

\bibliography{BCGrefs}
\end{document}